\documentclass[12pt,american,english]{article}
\usepackage[T1]{fontenc}
\usepackage[latin9]{inputenc}
\usepackage{amsmath}
\usepackage{amsthm}
\usepackage{amssymb}
\usepackage{geometry}
\usepackage{setspace}
\usepackage[authoryear]{natbib}
\makeatletter
\theoremstyle{definition}
 \newtheorem{example}{\protect\examplename}
\theoremstyle{definition}
\newtheorem{defn}{\protect\definitionname}
\theoremstyle{plain}
\newtheorem{thm}{\protect\theoremname}
\theoremstyle{plain}
\newtheorem{lem}{\protect\lemmaname}

\@ifundefined{date}{}{\date{}}

\usepackage{nicefrac}
\usepackage{setspace}
\usepackage{chngcntr}
\usepackage{microtype}
\usepackage{lmodern}
\usepackage{tikz}
\usepackage{tkz-tab}
\usepackage{caption}
\usepackage{subcaption}
\usepackage{graphicx}

\usepackage{abstract}

\usepackage{parskip}

\title{Prior-Free Blackwell}
\usepackage{amsmath,amsthm,amssymb,amsfonts} 

\author{Maxwell Rosenthal\thanks{Georgia Institute of Technology.  Email: \href{mailto:rosenthal@gatech.edu}{rosenthal@gatech.edu.}}}

\date{\today}

\def\D{\;\mathrm{d}}

\usepackage[colorlinks=true,linkcolor=blue,citecolor=blue,urlcolor=blue,hyperfootnotes=false]{hyperref}

\makeatother

\usepackage{babel}
\addto\captionsamerican{\renewcommand{\definitionname}{Definition}}
\addto\captionsamerican{\renewcommand{\examplename}{Example}}
\addto\captionsamerican{\renewcommand{\lemmaname}{Lemma}}
\addto\captionsamerican{\renewcommand{\theoremname}{Theorem}}
\addto\captionsenglish{\renewcommand{\definitionname}{Definition}}
\addto\captionsenglish{\renewcommand{\examplename}{Example}}
\addto\captionsenglish{\renewcommand{\lemmaname}{Lemma}}
\addto\captionsenglish{\renewcommand{\theoremname}{Theorem}}
\providecommand{\definitionname}{Definition}
\providecommand{\examplename}{Example}
\providecommand{\lemmaname}{Lemma}
\providecommand{\theoremname}{Theorem}

\begin{document}
\maketitle
\begin{abstract}
This paper develops a prior-free model of data-driven decision making
in which the decision maker observes the entire distribution of signals
generated by a known experiment under an unknown distribution of the
state variable and evaluates actions according to their worst-case
payoff over the set of state distributions consistent with that observation.
We show how our model applies to partial identification in econometrics
and propose a ranking of experiments in which $E$ is \emph{robustly
more informative} than $E'$ if the value of the decision maker's
problem after observing $E$ is always at least as high as the value
of the decision maker's problem after observing $E'$. This comparison,
which is strictly weaker than Blackwell's classical order, holds if
and only if the null space of $E$ is contained in the null space
of $E'$.
\end{abstract}

\section{Introduction}

A decision maker decides whether or not to assign binary treatment
$T$ to the entire population. The latent state of the world is $(Y_{0},Y_{1},T)$,
the data records the joint distribution $(Y_{T},T)$ of treatment
$T$ and realized outcome $Y_{T}$, and the decision maker evaluates
each action $a$ according to the worst-case expected outcome $Y_{a}$
over the set of latent distributions consistent with the data. What
are the consequences of uncertainty about the unobserved counterfactual
outcomes?

This paper develops a prior-free and empirically-motivated model of
decision making in which the decision maker does not know the prior
distribution $\mu$ of the state variable, observes the entire distribution
$E\mu$ of signals generated by known experiment $E$, and ranks actions
according to their worst-case expected utility over the set of state
distributions $\nu$ satisfying $E\nu=E\mu$. We propose a partial
ranking of experiments in which $E$ is \emph{robustly more informative
}than $E'$ if for every decision problem the decision maker's maxmin
payoff after observing the distribution of signals generated by the
former is at least her maxmin payoff after observing the distribution
of signals generated by the latter. 

In our main result, we show that $E$ is robustly more informative
than $E'$ if and only if the null space of $E$ is contained in the
null space of $E'$. Economically, the set of beliefs consistent with
observation of $E$ is a subset of the set of beliefs consistent with
observation of $E'$. Algebraically, there exists a matrix $\Gamma$
such that $E'=\Gamma E$. Because $E$ is Blackwell more informative
than $E'$ if and only if there exists a garbling matrix $G$ such
that $E'=GE$, it follows immediately that $E$ is robustly more informative
than $E'$ whenever the former is Blackwell more informative than
the latter. We show by example that the converse does not hold; in
doing so, we establish that the robust order admits strictly more
comparable pairs of experiments than the classical order.

The paper is organized as follows. First, we place the paper in the
literature in Section \ref{sec: lit}. Next, we lay out the model
and formalize the econometrics application in Section \ref{sec: model}.
We establish our main characterization in Section \ref{sec: main result}
and the connection to Blackwell in Section \ref{Section: robust weak}.
Finally, we conclude in Section \ref{sec: concl}. All proofs are
in the body.

\section{\label{sec: lit}Contribution and literature}

This paper has two goals. First, we provide an empirically-motivated
and prior-free model of experimentation. We view this framework as
a starting point for analyzing robust decision making in environments
with broad uncertainty about the distribution of payoff-relevant state
variables and an opportunity to learn about that distribution from
data. Our model encompasses partially-identified econometric analysis
as a special case, and our decision maker's worst-case criterion is
consistent with practice in that literature (\citet{Manski2003,Tamer2010,Molinari2020}).
We are distinguished by our interpretation of the identified set as
the output of an abstract signaling device, and by our explicit action--state
formulation of the decision problem.

Second, we develop a novel criterion for ranking experiments on the
basis of their value to a maxmin expected utility decision maker whose
belief set is completely characterized by the output of the experiment.
Our order is implied by but does not imply its classical counterpart
(\citet{blackwell1951comparison,blackwell1953equivalent}), and thus
provides a complementary tool for ranking experiments that are not
Blackwell-comparable. 

We contribute to a broader literature studying alternative models
of Blackwell experimentation. One strand of papers (\citet{celen2012informativeness,heyen2015informativeness,LI201618})
extends the original order to ambiguity-averse decision makers who
update a set of priors using a single draw from the experiment. Other
recent work includes \citet{wang2024informativeness}, who studies
Bayesian updating where the experiment itself is unknown, and \citet{Whitmeyer2025},
who provides a general treatment of non-Bayesian updating in experimental
contexts. Our paper differs in two important respects. First, our
decision maker observes the entire distribution of experimental output
rather than a single draw. Second, our decision maker learns by restricting
her belief set rather than updating her priors. To the best of our
knowledge, both features are novel.

\section{\label{sec: model}Model}

This paper takes a linear-algebraic approach to partial identification.
For any matrix $A$, we write $\text{null}(A)\equiv\{x\vert Ax=0\}$
for its null space. Given a finite set $X$, we write $\Delta(X)$
for the set of probability distributions on $X$ and identify $\Delta(X)$
with the corresponding subset of $\mathbb{R}^{\vert X\vert}$ equipped
with the Euclidean topology. More generally, given a compact metric
space $X$, we write $\Delta(X)$ for the set of Borel probabilities
on $X$ equipped with the topology of weak convergence.

\paragraph*{States and experiments}

The set of states $\Omega\equiv\{\omega_{1},...,\omega_{n}\}$ is
finite and the state of the world is distributed according to unknown
distribution $\mu$ in $\Delta(\Omega)$. We allow the decision maker
structural knowledge about the data generating process, and require
only that the set of state distributions $\mathcal{P}\subset\Delta(\Omega)$
that she entertains is compact and contains $\mu$. In order to allow
$\mathcal{P}$ to encode conditional independence assumptions of the
form $(Y_{0},Y_{1})\perp T\mid X$, we do not require that it is convex.

An \emph{experiment }is a finite set of messages $\{\sigma_{1},...,\sigma_{m}\}$
and an $m\times n$ matrix $E$ with entries 
\[
E_{ij}\equiv\text{Pr}(\sigma=\sigma_{i}\mid\omega=\omega_{j}).
\]
The number of messages $m$ is variable, the entries of $E$ are non-negative,
and the columns of $E$ sum to $1$. We suppress $\{\sigma_{1},...,\sigma_{m}\}$
wherever possible, identify experiments with their transition matrices
$E$, and make use of the fact that
\[
\forall d\in\text{null}(E)\;\;0=\sum_{i}\sum_{j}E_{ij}d_{j}=\sum_{j}d_{j}\sum_{i}E_{ij}=\sum_{j}d_{j}.
\]

\paragraph*{Actions and payoffs}

The set of actions $A$ is a non-empty and compact metric space and
the decision maker's utility $u:A\times\Omega\to\mathbb{R}$ is continuous.
We permit but do not require randomization, and remind the reader
that if $X$ is a compact metric space and $v:X\times\Omega\to\mathbb{R}$
is continuous then the set of mixed actions $A\equiv\Delta(X)$ on
$X$ is a compact metric space and the utility function $u:A\times\Omega\to\mathbb{R}$
defined by $u(a,\omega)\equiv\int_{X}v(x,\omega)\,\D a(x)$ is continuous.

\paragraph*{Beliefs}

The decision maker observes the distribution of messages $E\mu$ generated
by known experiment $E$ under the unknown state distribution $\mu$
and views state distribution $\nu$ in $\mathcal{P}$ as plausible
if and only if $E\nu=E\mu$. We write
\[
\mathcal{P}(E,\mu)\equiv\{\nu\in\mathcal{P}\mid E\nu=E\mu\}
\]
for the \emph{identified set} and make repeated use of the identity
\[
\{\nu\in\mathcal{P}\mid E\nu=E\mu\}=\{\nu\in\mathcal{P}\mid(\nu-\mu)\in\text{null}(E)\}.
\]
There are two extreme cases. First, if the columns of $E$ are linearly
independent, then the decision maker's model is exactly identified
and $\mathcal{P}(E,\mu)=\{\mu\}$. Second, if the columns of $E$
are identical, then the experiment conveys no information about the
true distribution of the state of the world and $\mathcal{P}(E,\mu)=\mathcal{P}$.

\paragraph*{The decision problem}

The decision maker evaluates actions according to their worst-case
expected payoff over the identified set $\mathcal{P}(E,\mu)$. Formally,
given experiment $E$, her problem is
\[
\max_{a\in A}\;\min_{\nu\in\mathcal{P}(E,\mu)}\;\sum_{j}u(a,\omega_{j})\nu(\omega_{j}).
\]
We call the tuple $(A,u,\mathcal{P},\mu)$ the \emph{decision problem}.
As a technical matter, because the map $(a,\nu)\mapsto\sum_{j}u(a,\omega_{j})\nu(\omega_{j})$
is continuous, the action set $A$ is compact, and the identified
set $\mathcal{P}(E,\mu)$ is a closed subset of a compact set $\mathcal{P}$
and therefore itself compact,
\begin{enumerate}
\item[(i)] the inner minimization problem has a solution $\nu^{*}$;
\item[(ii)] the worst-case expected utility $a\mapsto\min_{\nu\in\mathcal{P}(E,\mu)}\;\sum_{j}u(a,\omega_{j})\nu(\omega_{j})$
is continuous; and
\item[(iii)] the decision problem has a solution $a^{*}$.
\end{enumerate}

\paragraph*{Application}

Our framework encompasses robust causal inference in a variety of
partially identified econometric environments, with or without structural
restrictions. We illustrate here with a series of non-comprehensive
examples.

In our first two examples, the set $\mathcal{Y}\subset\mathbb{R}$
of outcomes $Y$, the set $\mathcal{X}$ of covariates $X$, and the
set $\mathcal{T}$ of treatments $T$ are each non-empty and finite.
The latent state $\omega\equiv((Y_{t})_{t},X,T)$ encodes the counterfactual
outcome $Y_{t}$ under treatment $t\in\mathcal{T}$, the covariate
$X$, and the factual treatment $T$. The set of latent states is
$\Omega\equiv\mathcal{Y}^{\mathcal{T}}\times\mathcal{X}\times\mathcal{T}$,
the set of actions $A$ coincides with the set of treatments $\mathcal{T}$,
and the decision maker's utility $u:A\times\Omega\to\mathbb{R}$ is
\[
u(a,((Y_{t})_{t},X,T))\equiv Y_{a}.
\]
The decision maker is not biased towards a treatment status quo, and
each action is evaluated with respect to the symmetric worst-case
expected outcome criterion.
\begin{example}
\label{example 1}The set of state distributions is $\mathcal{P}\equiv\Delta(\Omega)$,
the set of signals $\mathcal{Y}\times\mathcal{X}\times\mathcal{T}$
is the set of observable states $(Y_{T},X,T)$, and the experiment
$E$ sends signal $(Y,X,T)$ with probability $1$ in all latent states
$((Y_{t})_{t},X',T')$ with $(X',T')=(X,T)$ and $Y_{T}=Y$.
\end{example}
The decision maker in Example \ref{example 1} is unwilling to exogenously
impose the conditional independence assumption $(Y_{t})_{t}\perp T\mid X$
that identifies the causal mean outcomes $E[Y_{a}]$ from the marginal
distribution of observable variables $(Y_{T},X,T)$.
\begin{example}
\label{example 2}The set of state distributions $\mathcal{P}\equiv\{\nu\in\Delta(\Omega)\mid(Y_{t})_{t}\perp_{\nu}T\mid X\}$
reflects exogenous knowledge that treatment is unconfounded, the set
of signals $\mathcal{Y}\times\mathcal{T}$ is the set of outcome--treatment
pairs $(Y_{T},T)$, and experiment $E$ sends signal $(Y,T)$ with
probability $1$ in all latent states $((Y_{t})_{t},X,T')$ with $(Y_{T},T')=(Y,T)$. 
\end{example}
While the decision maker in Example \ref{example 2} does know that
treatment is conditionally independent of the vector of counterfactual
outcomes, she observes only the marginal distribution of observable
variables $(Y_{T},T)$ rather than the full joint distribution $(Y_{T},X,T)$.
Her model suffers from potential omitted variable bias and she hedges
against that uncertainty by adopting a robust decision making criterion. 

In our third example, the set $\mathcal{Y}\subset\mathbb{R}$ of outcomes
$Y$ is non-empty and finite, the set of treatments is $\mathcal{T}\equiv\{0,1\}$,
and the set of instruments is $\mathcal{Z}\equiv\{0,1\}$. The latent
state $\omega\equiv(Y_{0},Y_{1},T(0),T(1),Z)$ encodes the counterfactual
outcomes $(Y_{0},Y_{1})$ under the two treatments, the treatment
$T(0)$ taken under non-encouragement $Z=0$, the treatment $T(1)$
taken under encouragement $Z=1$, and the value of the instrument
$Z$. The set of latent states is $\Omega\equiv\mathcal{Y}^{\mathcal{T}}\times\mathcal{T}^{\mathcal{Z}}\times\mathcal{Z}$,
the set of actions $A$ coincides with the set of treatments $\mathcal{T}$,
and the decision maker's utility $u:A\times\Omega\to\mathbb{R}$ is
$u(a,(Y_{0},Y_{1},T(0),T(1),Z))\equiv Y_{a}$. 
\begin{example}
\label{example 3}The set of state distributions $\mathcal{P}\equiv\{\nu\in\Delta(\Omega)\mid Z\perp_{\nu}(Y_{0},Y_{1},T(0),T(1))\}$
reflects exogenous knowledge that the instrument is randomly assigned.
The set of signals $\mathcal{Y}\times\mathcal{T}\times\mathcal{Z}$
is the set of outcome--treatment--instrument triples $(Y_{T},T,Z)$,
and the experiment $E$ sends signal $(Y,T,Z)$ with probability $1$
in all latent states $(Y_{0},Y_{1},T(0),T(1),Z')$ with $(Y_{T(Z')},T(Z'),Z')=(Y,T,Z)$.
\end{example}
In Example \ref{example 3}, the decision maker knows that the instrument
is randomly assigned and she observes the full distribution of the
observables. Nevertheless, as in \citet*{ImbensAngrist1994,AngristImbensRubin1996},
the unconditional mean outcome $Y_{a}$ is unidentified even under
the additional monotonicity assumption $T(1)\geq T(0)$.

\section{\label{sec: main result}Robust informativeness}

The analytical goal of this paper is to rank experiments by their
value to the decision maker.
\begin{defn}
Experiment\emph{ $E$ is robustly more informative} than experiment
$E'$ if for all decision problems $(A,u,\mathcal{P},\mu)$ 
\[
\max_{a\in A}\;\min_{\nu\in\mathcal{P}(E,\mu)}\;\sum_{j}u(a,\omega_{j})\nu(\omega_{j})\geq\max_{a\in A}\;\min_{\nu\in\mathcal{P}(E',\mu)}\;\sum_{j}u(a,\omega_{j})\nu(\omega_{j}).
\]
One experiment is robustly more informative than another if the value
of the decision maker's problem after observing the output of the
former is always at least the value of that problem after observing
the latter. In our main result, we establish that this order is equivalent
to three other criteria.
\end{defn}
\begin{thm}
\label{Theorem}The following four statements are equivalent:
\begin{enumerate}
\item[(i)] there exists a matrix $\Gamma$ such that $E'=\Gamma E$; 
\item[(ii)] $\text{null}(E)\subset\text{null}(E')$; 
\item[(iii)] $\mathcal{P}(E,\mu)\subset\mathcal{P}(E',\mu)$ for all $\mathcal{P}$
and for all $\mu\in\mathcal{P}$;
\item[(iv)] experiment $E$ is robustly more informative than experiment $E'$.
\end{enumerate}
\end{thm}
As we show in Theorem \ref{Theorem}, experiment $E$ is robustly
more informative than $E'$ if and only if there exists a matrix $\Gamma$
such that $E'=\Gamma E$. In contrast, $E$ is Blackwell more informative
than $E'$ if and only if there exists a garbling matrix $G$ such
that the entries of $G$ are non-negative, the columns of $G$ sum
to $1$, and $E'=GE$. This implies that our order is no stronger
than its classical counterpart and suggests that it is weaker, as
we establish in Section \ref{Section: robust weak}. For now, we develop
the theorem as the summary of three lemmas. First, (i) and (ii) are
equivalent.
\begin{lem}
\label{Lemma 1}There exists a matrix $\Gamma$ with $E'=\Gamma E$
if and only if $\text{null}(E)\subset\text{null}(E')$. 
\end{lem}
\begin{proof}
If $E'=\Gamma E$ and $Ex=0$ then 
\[
E'x=\Gamma Ex=\Gamma0=0
\]
and hence $\text{null}(E)\subset\text{null}(E')$. Conversely, if
$\text{null}(E)\subset\text{null}(E')$ then the fundamental theorem
of linear algebra implies the rows of $E$ span the rows of $E'$
and thus there exists a matrix $\Gamma_{m'\times m}$ such that
\[
(E'_{i1},...,E'_{in})=\sum_{j=1}^{m}\Gamma_{ij}(E_{j1},...,E_{jn}).
\]
Consequently, $E'=\Gamma E$.
\end{proof}
Next, (ii) and (iii) are equivalent.
\begin{lem}
\label{Lemma 2}Experiments $E,E'$ satisfy $\text{null}(E)\subset\text{null}(E')$
if and only if $\mathcal{P}(E,\mu)\subset\mathcal{P}(E',\mu)$ for
all $\mathcal{P}$ and for all $\mu\in\mathcal{P}$.
\end{lem}
\begin{proof}
If $\text{null}(E)\subset\text{null}(E')$ then $\mathcal{P}(E,\mu)\subset\mathcal{P}(E',\mu)$
for all $\mathcal{P}$ and for all $\mu$. Conversely, if there exists
$d\in\text{null}(E)\setminus\text{null}(E')$ then let $\mu$ have
full support, choose $\lambda\in\mathbb{R}\setminus\{0\}$ with $\vert\lambda\vert$
sufficiently small to satisfy $\nu\equiv(\mu+\lambda d)\in\Delta(\Omega)$,
and let $\mathcal{P}\subset\Delta(\Omega)$ be compact and contain
both $\mu$ and $\nu$. Because $(\nu-\mu)\in\text{null}(E)\setminus\text{null}(E')$,
we have $\nu\in\mathcal{P}(E,\mu)\setminus\mathcal{P}(E',\mu)$. 
\end{proof}
Finally, (iii) and (iv) are substantively equivalent. 
\begin{lem}
\label{Lemma 3}If $\mathcal{P}(E,\mu)\not\subset\mathcal{P}(E',\mu)$
then there exists an action set $A$ and a utility function $u$ such
that
\[
\max_{a\in A}\;\min_{\nu\in\mathcal{P}(E,\mu)}\;\sum_{j}u(a,\omega_{j})\nu(\omega_{j})<\max_{a\in A}\;\min_{\nu\in\mathcal{P}(E',\mu)}\;\sum_{j}u(a,\omega_{j})\nu(\omega_{j})
\]
.
\end{lem}
\begin{proof}
Let $\nu\in\mathcal{P}(E,\mu)\setminus\mathcal{P}(E',\mu)$ and interpret
sets $S\equiv\{\nu\},T\equiv\text{conv}(\mathcal{P}(E',\mu))$ as
subsets of $\mathbb{R}^{\vert\Omega\vert}$. First, because $E'$
is linear and $E'\nu'=E'\mu$ for all $\nu'\in\mathcal{P}(E',\mu)$,
we have $E'\nu'=E'\mu$ for all $\nu'\in T\supset\mathcal{P}(E',\mu)$.
In turn, because $E'\nu\neq E'\mu$, we conclude $S$ and $T$ are
disjoint. Second, because $S$ and $T$ are also compact and convex,
the separating hyperplane theorem provides a vector $\alpha\in\mathbb{R}^{n}$
and a constant $\beta\in\mathbb{R}$ such that
\[
\forall\nu'\in T\;\;\sum_{j}\alpha_{j}\nu(\omega_{j})<\beta<\sum_{j}\alpha_{j}\nu'(\omega_{j}).
\]
Third, because $\mathcal{P}(E',\mu)\subset T$ we have
\[
\min_{\nu'\in\mathcal{P}(E,\mu)}\sum_{j}\alpha_{j}\nu'(\omega_{j})\leq\sum_{j}\alpha_{j}\nu(\omega_{j})<\min_{\nu'\in T}\sum_{j}\alpha_{j}\nu'(\omega_{j})\leq\min_{\nu'\in\mathcal{P}(E',\mu)}\sum_{j}\alpha_{j}\nu'(\omega_{j}).
\]
Set $A\equiv\{a\}$ and $u(a,\omega_{j})\equiv\alpha_{j}$.
\end{proof}
Theorem \ref{Theorem} now follows quickly from Lemmas \ref{Lemma 1}--\ref{Lemma 3}.
\begin{proof}[Proof of Theorem \ref{Theorem}]
Lemma \ref{Lemma 1} and Lemma \ref{Lemma 2} imply (i), (ii), and
(iii) are equivalent, and it is apparent that (iii) implies (iv).
By Lemma \ref{Lemma 3}, the negation of (iii) yields a decision problem
$(A,u,\mathcal{P},\mu)$ such that
\[
\max_{a\in A}\;\min_{\nu\in\mathcal{P}(E,\mu)}\;\sum_{j}u(a,\omega_{j})\nu(\omega_{j})<\max_{a\in A}\;\min_{\nu\in\mathcal{P}(E',\mu)}\;\sum_{j}u(a,\omega_{j})\nu(\omega_{j}).
\]
Accordingly, the negation of (iii) implies the negation of (iv).
\end{proof}

\section{\label{Section: robust weak}Blackwell implies robust informativeness}

In this section, we formalize our earlier statements about the relationship
between robust informativeness and Blackwell informativeness.
\begin{defn}
\label{Definition: blackwell}Experiment $E$ is \emph{Blackwell more
informative} than $E'$ if there exists a matrix $G$ such that the
entries of $G$ are non-negative, the columns of $G$ sum to $1$,
and $E'=GE$.
\end{defn}
The additional restrictions on $G$ strengthen Blackwell's order relative
to our own.
\begin{thm}
\label{thm 2}If $E$ is Blackwell more informative than $E'$ then
$E$ is robustly more informative than $E'$.
\end{thm}
\begin{proof}
The claim follows immediately from the equivalence between characterization
(i) and (iv) in the statement of Theorem \ref{Theorem} and from Definition
\ref{Definition: blackwell}.
\end{proof}
Although Theorem \ref{Theorem} and Theorem \ref{thm 2} do not on
their own imply that our order is strictly weaker than Blackwell's,
we confirm by example that there are robustly comparable experiments
which are Blackwell incomparable.
\begin{example}
Let experiments $E,E'$ both send messages $\{\sigma_{1},\sigma_{2},\sigma_{3}\}$
with transition matrices
\begin{align*}
E & \equiv\begin{pmatrix}1/2 & 0 & 0\\
1/2 & 1/2 & 0\\
0 & 1/2 & 1
\end{pmatrix}, & E' & \equiv\begin{pmatrix}1 & 0 & 0\\
0 & 1 & 1\\
0 & 0 & 0
\end{pmatrix}
\end{align*}
and define

\begin{align*}
\Gamma & \equiv\begin{pmatrix}2 & 0 & 0\\
-1 & 1 & 1\\
0 & 0 & 0
\end{pmatrix}.
\end{align*}
First, because the columns of $E$ are linearly independent and the
columns of $E'$ are linearly dependent, the null space of $E$ is
a proper subset of the null space of $E$'. In turn, Theorem \ref{Theorem}
implies that $E$ is strictly robustly more informative than $E'$
and Theorem \ref{thm 2} implies that $E'$ is not weakly Blackwell
more informative than $E$. Second, because $E$ is invertible, $\Gamma$
is the only real matrix $G$ satisfying $E'=GE$. Because $\Gamma$
is not a garbling, $E$ is not Blackwell more informative than $E'$.
While any experiment that exactly identifies the state distribution
$\mu$ is a maximum element in the robust informativeness order, it
is not necessarily Blackwell more informative than experiments that
do not point identify $\mu$.
\end{example}

\section{\label{sec: concl}Conclusions}

This paper introduces a fully prior-free model of data-driven decision
making and identifies applications to partially-identified econometric
decision problems. We propose a partial ranking of experiments based
on their value to a decision maker with no prior information about
the distribution of the state variable, and provide a sharp characterization
via null-space inclusion. Our order is implied by but does not imply
Blackwell's, and thus offers a richer set of comparisons than its
classical counterpart. While we have written the paper for the maxmin
expected utility criterion, our results generalize readily to any
ambiguity-averse preference under which payoffs increase monotonically
as the uncertainty set shrinks, as implied by the third characterization
in Theorem \ref{Theorem}. We leave the formalization of this observation
to future work.

\selectlanguage{american}%
\bibliographystyle{chicago}
\bibliography{rb}

@article{blackwell1951comparison,
  author = {Blackwell, David},
  title = {Comparison of Experiments},
  journal = {Proceedings of the Second Berkeley Symposium on Mathematical Statistics and Probability},
  volume = {1},
  pages = {11--22},
  year = {1951},
  publisher = {University of California Press},
  address = {Berkeley, CA},
  url = {https://projecteuclid.org/euclid.bsmsp/1200500216}
}

@article{blackwell1953equivalent,
  author = {Blackwell, David},
  title = {Equivalent Comparisons of Experiments},
  journal = {Annals of Mathematical Statistics},
  volume = {24},
  number = {2},
  pages = {265--272},
  year = {1953},
  doi = {10.1214/aoms/1177729433},
  url = {https://doi.org/10.1214/aoms/1177729433}
}

@article{wang2024informativeness,
  author = {Wang, Zichang},
  title = {Informativeness orders over ambiguous experiments},
  journal = {Journal of Economic Theory},
  volume = {213},
  pages = {105937},
  year = {2024},
  doi = {10.1016/j.jet.2024.105937},
  url = {https://doi.org/10.1016/j.jet.2024.105937},
  publisher = {Elsevier}
}

@article{celen2012informativeness,
  author = {{\c{C}}elen, Bo{\u{g}}a{\c{c}}han},
  title = {Informativeness of experiments for MEU},
  journal = {Journal of Mathematical Economics},
  volume = {48},
  number = {6},
  pages = {404--406},
  year = {2012},
  doi = {10.1016/j.jmateco.2012.08.007},
  url = {https://doi.org/10.1016/j.jmateco.2012.08.007},
  publisher = {Elsevier}
}

@article{heyen2015informativeness,
  author = {Heyen, Daniel and Wiesenfarth, Boris R.},
  title = {Informativeness of experiments for MEU---A recursive definition},
  journal = {Journal of Mathematical Economics},
  volume = {57},
  pages = {71--77},
  year = {2015},
  doi = {10.1016/j.jmateco.2014.12.002},
  url = {https://doi.org/10.1016/j.jmateco.2014.12.002},
  publisher = {Elsevier}
}

@article{LI201618,
title = {Blackwell's informativeness ranking with uncertainty-averse preferences},
journal = {Games and Economic Behavior},
volume = {96},
pages = {18-29},
year = {2016},
issn = {0899-8256},
doi = {https://doi.org/10.1016/j.geb.2016.01.009},
url = {https://www.sciencedirect.com/science/article/pii/S0899825616000130},
author = {Jian Li and Junjie Zhou}
}

@article{Whitmeyer2025,
  author       = {Whitmeyer, Mark},
  year         = {2025},
  title        = {Blackwell-Monotone Updating Rules},
  journal      = {Journal of Political Economy},
  note         = {Forthcoming}
}

@book{Manski2003,
  author       = {Manski, Charles F.},
  title        = {Partial Identification of Probability Distributions},
  year         = {2003},
  publisher    = {Springer},
  address      = {New York},
}

@article{Tamer2010,
  author       = {Tamer, Elie},
  title        = {Partial Identification in Econometrics},
  journal      = {Annual Review of Economics},
  volume       = {2},
  number       = {1},
  pages        = {167--195},
  year         = {2010},
}

@incollection{Molinari2020,
  author    = {Molinari, Francesca},
  title     = {Microeconometrics with Partial Identification},
  booktitle = {Handbook of Econometrics},
  series    = {Handbook of Econometrics},
  volume    = {7A},
  pages     = {355--486},
  year      = {2020},
  publisher = {Elsevier},
  editor    = {Durlauf, Steven N. and Hansen, Lars Peter and Heckman, James J. and Matzkin, Rosa L.},
  doi       = {10.1016/bs.hoe.2020.05.002},
  isbn      = {978-0-444-63649-2},
}

@article{ImbensAngrist1994,
  title={Identification and Estimation of Local Average Treatment Effects},
  author={Imbens, Guido W. and Angrist, Joshua D.},
  journal={Econometrica},
  volume={62},
  number={2},
  pages={467--475},
  year={1994},
  publisher={Econometric Society}
}

@article{AngristImbensRubin1996,
  title={Identification of Causal Effects Using Instrumental Variables},
  author={Angrist, Joshua D. and Imbens, Guido W. and Rubin, Donald B.},
  journal={Journal of the American Statistical Association},
  volume={91},
  number={434},
  pages={444--455},
  year={1996},
  publisher={Taylor \& Francis}
}
\selectlanguage{english}%

\end{document}